\documentclass[12pt, a4paper, fleqn]{article}

\usepackage{etex}
\usepackage[utf8]{inputenc}
\usepackage[TS1,T1]{fontenc}
\usepackage{xspace}

\usepackage{array}
\newcolumntype{P}[1]{>{\centering\arraybackslash}p{#1}}
\usepackage{multirow}
\usepackage{booktabs, dcolumn}
\usepackage{tabularx}

\usepackage{tikz-cd}
\usetikzlibrary{external}
\usepackage{smartdiagram}
\usepackage{tensor}
\usetikzlibrary{decorations.markings}
\usetikzlibrary{arrows.meta}
\usetikzlibrary{arrows,matrix,positioning,shapes,arrows}
\usetikzlibrary{shapes.geometric, arrows, calc, intersections}

\newcommand{\tikznode}[2]{\relax
	\ifmmode%
	\tikz[remember picture,baseline=(#1.base),inner sep=0pt] \node (#1) {$#2$};
	\else
	\tikz[remember picture,baseline=(#1.base),inner sep=0pt] \node (#1) {#2};%
	\fi}

\usepackage{layout}
\usepackage{pdfpages}
\usepackage{float}

\usepackage[a4paper,tmargin=2.75truecm,bmargin=2.5truecm,hmargin=2.8truecm, 
rmargin=2.9truecm, lmargin=2.9truecm, marginparwidth=25mm]{geometry}
\usepackage{tocloft}
\usepackage[english]{babel}
\usepackage{blindtext} %This package generates automatic text
\usepackage{epigraph}
\usepackage[]{libertinus-type1}
\usepackage{amsfonts, mathtools,slashed,mathdots,bm}
\usepackage{libertinust1math}
\usepackage[bb=ams]{mathalpha}

\makeatletter
\newcommand{\addbar@}[3]{%
	\makebox[0pt][l]{%
		\raisebox{#1}[0pt][0pt]{%
			\kern#2 \scalebox{#3}[0.8]{$\m@th\mathchar"84$}%
		}%
	}%
}
\DeclareRobustCommand{\lambdabar}{\text{\addbar@{0.1ex}{0.18em}{1}}\lambda}
\makeatother

\usepackage{enumitem}

\newlist{enum-hypothesis}{enumerate}{1}
\setlist[enum-hypothesis]{label=(\arabic*),itemsep=0pt, parsep=0pt}

\setlist[enumerate,1]{label=\arabic*., ref=\arabic*, topsep=1pt, itemsep=2pt, parsep=0pt, 
	leftmargin=1.5em, itemindent=0em, labelsep=0.2em, labelwidth=1.3em}

\setlist[enumerate,2]{label=\alph*., ref=\theenumi.\alph*, topsep=1pt, itemsep=2pt, parsep=0pt, 
	leftmargin=0.5em, itemindent=0em, labelsep=0.2em, labelwidth=1.5em}

\setlist[enumerate,3]{label=\roman*., ref=\theenumii.\roman*, topsep=1pt, itemsep=2pt, parsep=0pt, 
	leftmargin=0.5em, itemindent=0em, labelsep=0.2em, labelwidth=1.2em}

\usepackage{graphicx, subfig}
\usepackage{pbox}

\usepackage[square,numbers,compress,merge]{natbib}
\usepackage[thmmarks,amsmath]{ntheorem}

\newtheorem{theorem}{Theorem}[section]
\newtheorem{proposition}[theorem]{Proposition}

\newtheorem{definition}[theorem]{Definition}

\theoremsymbol{$\square$}
\theoremstyle{plain}
\theorembodyfont{\upshape}

\theoremsymbol{$\Diamond$}

\theoremsymbol{}
\theoremstyle{break}
\theorembodyfont{\slshape}

\theoremheaderfont{\scshape}
\theorembodyfont{\upshape}
\theoremstyle{nonumberplain}
\theoremseparator{}
\theoremsymbol{\rule{1.3ex}{1.3ex}}
\newtheorem{proof}{Proof}

\usepackage{hyperref}
\hypersetup{
	plainpages=false,
	colorlinks=true,% supprime l'encadrement des liens
	linkcolor=blue,
	anchorcolor=black,
	citecolor=blue,
	urlcolor=blue,
	menucolor=black,
	filecolor=black,
	bookmarksopen=true,
	bookmarksnumbered=true}

\hypersetup{
	pdfauthor={Gaston Nieuviarts},
	pdftitle={Titre},
	pdfsubject={},
	pdfcreator={pdflatex},
	pdfproducer={pdflatex},
	pdfkeywords={}
}

\newcommand\bbZ{\mathbb{Z}}

\newcommand{\bbbone}{{\text{\usefont{U}{bbold}{m}{n}\char49}}}

\newcommand{\fbb}{\bbbone_F}
\newcommand{\mbb}{\bbbone_{2^{m}}}

\newcommand{\Man}{\mathcal{M}}
\newcommand{\CM}{C^\infty(\Man)}

\newcommand{\calA}{\mathcal{A}}
\newcommand{\calB}{\mathcal{B}}

\newcommand{\calH}{\mathcal{H}}

\newcommand{\calK}{\mathcal{K}}

\newcommand{\calO}{\mathcal{O}}

\newcommand{\calU}{\mathcal{U}}

\newcommand{\algA}{\calA}

\newcommand{\hs}{\calH}

\newcommand{\HF}{\calH_F}

\newcommand{\Ths}[1][]{\widetilde{\hs}}

\newcommand{\FDir}{\Dir_F}

\newcommand{\tw}{K}

\newcommand{\twx}{\tw_x}

\newcommand{\KDir}{\Dir^\tw}

\newcommand{\inn}{{\langle\, \cdot \, , \, \cdot \, \rangle}}
\newcommand{\inntw}{{\langle\, \cdot \, , \, \cdot \, \rangle_\tw}}

\newcommand{\Tadj}{{\dagger_{\tiny \tw}}}

\newcommand{\JF}{J_F}
\newcommand{\JH}{{\hat{J}}}

\newcommand{\kgm}{\gamma_\tw}
\newcommand{\tgm}{\tilde\gamma}
\newcommand{\GF}{\Gamma_F}

\newcommand{\AF}{\calA_F}

\newcommand{\g}{g}

\newcommand{\gr}{\g_{\scriptscriptstyle R}}
\newcommand{\gk}{\g_{\scriptscriptstyle \tw}}

\newcommand{\tm}{{T\Man}}

\DeclareMathOperator{\Aut}{Aut}
\DeclareMathOperator{\Inn}{Inn}

\newcommand{\gl}{{GL}}
\newcommand{\glv}{\gl(V)}

\newcommand{\US}{S} %Unit sphere
\newcommand{\CH}{{\hat{C}}}

\newcommand{\Sp}{\mathcal{S}} % Spin structure
\newcommand{\Dir}{D}
\newcommand{\defeq}{\vcentcolon=} % :=
\DeclareMathOperator{\Tr}{Tr} %% trace

\newcommand{\twphi}{{ \phi^{\,\tw}}}

\newcommand{\cl}{c}
\newcommand{\tcl}{\tilde{c}}

\newcounter{mnotecount}[section]
\renewcommand{\themnotecount}{\thesection.\arabic{mnotecount}}

\newcommand{\mnote}[1]%
{\protect{\stepcounter{mnotecount}}${}^{\text{\footnotesize$\bullet$\themnotecount}}$%
	\reversemarginpar%
	\marginpar{\raggedleft\footnotesize$\bullet$\themnotecount: #1}}

\newlength{\mnotewidth}
\definecolor{blueamu}{RGB}{0, 101, 189}
\definecolor{cyanamu}{RGB}{61, 183, 228}

\newcommand{\dhorline}[3][0]{%
	\tikz[baseline=-2pt]{\path[decoration={markings, 
			mark=between positions 0 and 1 step 2*#3 with 
			{\node[color=blueamu, fill, circle, minimum width=#3, inner sep=0pt, anchor=south west] {};}},
		postaction={decorate}] (0,#1) -- ++(#2,0);}}

\newcommand{\dvertline}[3][0]{%
	\tikz[baseline=2em]{\path[decoration={markings, 
			mark=between positions 0 and 1 step 2*#2 with 
			{\node[color=black!50, fill, circle, minimum width=#2, inner sep=0pt, anchor=south west] {};}},
		postaction={decorate}] (0, #1) -- ++(0,#3);}}

\makeatletter\newcommand\HUGE{\@setfontsize\Huge{28}{0}}\makeatother

\numberwithin{equation}{section}
\allowdisplaybreaks
\title{Emergence of Time\\ from a Twisted Spectral Triple\\ in Almost-Commutative Geometry}

\author{Gaston Nieuviarts\footnote{gaston.nieuviarts.wk@gmail.com}}

\date{October 2025}

\begin{document}
\maketitle

\vspace{-0.99cm}
 
\begin{abstract}
These proceedings present a synthesis of recent results on the emergence of pseudo-Riemannian structures from twisted spectral triples within the almost-commutative framework. It provides an algebraic mechanism for addressing the Lorentzian signature problem, demonstrating how the almost-commutative structure underlying the noncommutative Standard Model of particle physics gives rise to a Lorentzian spectral triple from a purely Riemannian setting. This notably offers an alternative to Wick rotation, provided by a notion of morphism connecting twisted and pseudo-Riemannian spectral triples. This remains a local result in the compact setting, rather than a full Lorentzian space-time with global causal structure.
 
\end{abstract}

\section*{Introduction}

Noncommutative geometry (NCG) has provided a powerful algebraic reformulation of Riemannian spin geometry in purely spectral terms: a spectral triple $(\calA, \calH, \Dir)$ encodes the full metric content of a space through the Dirac operator $\Dir$ and the representation of $\calA$ on $\calH$. In the case of a (compact) manifold $\Man$, the underlying algebra is the commutative algebra $\calA=\CM$ of smooth functions over $\Man$, see \cite{connes2013spectral,connes1996gravity, connes2006quantum} for more information. 

An interesting minimal noncommutative extension concerns spectral triples based on the so-called almost-commutative algebras $\CM\otimes\AF$ where $\AF$ is a finite-dimensional algebra. This structure offers a compelling framework for particle physics, allowing the geometrization of Riemannian gravity and fundamental interactions within a unified picture. This framework couples the classical manifold to a finite internal algebra and thereby packages gauge fields and fermions into geometry itself. Inner fluctuations of $\Dir$ generate the gauge bosons and the Higgs field. The spectral action, in turn, produces the bosonic Lagrangian and ties particle physics to geometry within a single unified framework. In this way, the (Euclidean) Standard Model with neutrino mixing arises from an almost-commutative geometry, see \cite{ConnLott90a, chamseddine1997spectral, chamseddine2007gravity} for more details on this construction.

Yet, in its standard formulation, NCG is essentially Euclidean: its axioms are built on a Hilbert space and a positive-definite inner product, so that time and Lorentzian symmetries do not enter the spectral data. For realistic physical models, most notably the noncommutative Standard Model, one requires a genuine notion of time and a consistent implementation of Lorentz covariance within the spectral framework. Addressing the signature problem in NCG is therefore both conceptually central and algebraically delicate. Different contributions tackling this issue can be found in \cite{barrett2007lorentzian, franco2014temporal, strohmaier2006noncommutative, devastato2018lorentz, d2016wick, van2016krein, bizi2018space}.

The problem can be stated as follows: mathematics tells us that only Riemannian geometry is natural within NCG's framework, whereas physics tells us that pseudo-Riemannian geometry is real. 
The present proceeding offers a synthesis of the approach developed in \cite{nieuviarts2025emergence, nieuvsignchange}, addressing this problem by proposing a Riemannian-like framework (Hilbert space, self-adjoint Dirac operator, etc.) from which pseudo-Riemannian structures and symmetries emerge naturally within the almost-commutative picture. 

The approach starts from the Riemannian side but equips spectral triples with a twist $\rho$ in the sense of Connes–Moscovici (see \cite{connes2006type} and Section~\ref{TwistedSec}).  
An important result of \cite{nieuviarts2025emergence, nieuvsignchange}, presented in Section~\ref{SecMorphism}, was to show how the so-called pseudo-Riemannian spectral triple introduced in Section~\ref{SecPseudoRiemST}, can be related to such twisted spectral triples through a morphism. This morphism implements a local change of signature, without introducing any complex numbers, via a reflection operator $r$ directly derived from the chosen twist. In this framework, the twist naturally plays the role of the parity operator.

The main result, discussed in Section~\ref{SecACTW}, demonstrates that extending the previously defined twisted spectral triple with an almost-commutative structure induces the emergence of the corresponding pseudo-Riemannian spectral triple with its Krein product, as a consequence of the algebraic constraints imposed by the almost-commutative spectral triple on the full Dirac operator.  
Here, the twist appears as a natural ingredient, gluing the commutative and finite parts at the level of the fluctuations, thus establishing an interesting connection between twisted and almost-commutative spectral triple structures.

This mechanism makes precise how time emerges from the twisted spectral data, i.e., from the chosen twist $\rho$ within the almost-commutative spectral triple framework underlying the noncommutative Standard Model. This provides a controlled and algebraically transparent pathway from Euclidean NCG to Lorentzian physics within the spectral paradigm. It refines the prospects for a Lorentzian formulation of the noncommutative Standard Model and offers a conceptual alternative to Wick rotation.

\section{Twisted Spectral Triples}
\label{TwistedSec}
The following section presents the essentials on twisted spectral triples. A particular focus is done on the associated product, the corresponding unitary notion, their relation with twisted fluctuations and the real structure. More details can be found in \cite{nieuvsignchange, nieuviarts2025emergence}.
 
We recall the definition of a spectral triple. A spectral triple $(\algA, \calH,\Dir)$ is the data of an involutive unital algebra
$\algA$ represented by bounded operators on a Hilbert space $\calH$, and of a self-adjoint operator $\Dir$ acting on $\calH$ such that the resolvent $(i + \Dir^2)^{-1}$ is compact and that for any $a\in\algA$, $[\Dir, a]$ is a bounded operator. A real spectral triple $(\algA, \hs, \Dir, J)$ is defined by the introduction of the operator $J$ in the spectral triple, where $J$ is an anti-unitary operator so that
\begin{equation}
	\label{EqDefStructParam}
	J^2 = \epsilon_0,\qquad\qquad\quad J \Dir = \epsilon_1 \Dir J,\qquad\qquad\qquad \epsilon_0, \epsilon_1 \in \{\pm 1\}.
\end{equation}
We refer to \cite{connes2013spectral,connes1996gravity, connes2006quantum} for more details on spectral triples axioms.

Spectral triples have been successfully defined for Type~I and Type~II von~Neumann algebras, where the existence of a faithful trace allows for a natural formulation of the spectral data. However, Type~III algebras, appearing in quantum field theory for example, pose a significant challenge due to the absence of such a trace. To address this difficulty, A.~Connes and H.~Moscovici introduced in~\cite{connes2006type} the concept of \emph{twisted spectral triples}, where an automorphism $\rho$ called twist is incorporated into the definition so as to compensate for the lack of a trace and restore a viable spectral framework. Twisted spectral triples were then studied in different contexts, for instance by A. Devastato, S. Farnsworth, M. Filaci, G. Landi, F. Lizzi, P. Martinetti, D. Singh, R. Ponge see \cite{TwistGaugeLandiMarti2018, martinetti2022lorentzian, TwistLandiMarti2016, ponge2016index, filaci2023critical}.

A twisted spectral triple $(\calA, \calH, \Dir, \rho)$ is defined as follows. One introduces an automorphism $\rho \in \Aut(\calA)$ satisfying the regularity condition
\begin{equation}
	\rho(a^{\dagger})=(\rho^{-1}(a))^{\dagger}\qquad \qquad \qquad \forall a\in\calA.
\end{equation}
The usual commutator $[\Dir,a]$ is replaced by the twisted commutator
\begin{equation}
	[\Dir,a]_\rho:= \Dir a-\rho(a)\Dir,
\end{equation}
and only $[\Dir,a]_\rho$ is required to be bounded. The specifications on $\calA,\calH$ and $\Dir$ are the same as for spectral triples. The axioms of real twisted spectral triples $(\calA, \calH, \Dir, J, \rho)$ coincide with those of ordinary spectral triples to the exception of the replacement of the first-order condition by the twisted first-order condition:
\begin{equation}
	[[\Dir, a]_\rho, b^\circ]_{\rho^\circ}=0\qquad\qquad\qquad \forall a\in\calA, \,\, b^\circ \in\calA^\circ,
\end{equation}
where $\calA^\circ$ is the opposite algebra, and $\rho^\circ$ acting on $\calA^\circ$ is defined by $\rho^\circ(a^\circ):= (\rho^{-1}(a))^\circ$.
 
Derivations are replaced by twisted derivations $\delta_\rho (a):= [\Dir,a]_\rho$, and the corresponding set of twisted $1$-forms is defined by
\begin{equation}
	\Omega^1_\Dir(\calA,\rho):=\{\, \sum_ia_i[\Dir, b_i]_{\rho}\,\, \mid\,\, a_i, b_i \in \calA\, \}.
\end{equation}
The set $\Omega^1_\Dir(\calA, \rho)$ forms an $\calA$-bimodule for the following product
\begin{equation}
	\label{EqLR}
	a\, \cdot \, A_\rho \, \cdot \, b=\rho(a)A_\rho b,
\end{equation}
for which the twisted derivation satisfies
\begin{equation}
	\delta_\rho(ab)=\delta_\rho(a)b+\rho(a)\delta_\rho(b)=\delta_\rho(a)\, \cdot \, b+a \, \cdot \, \delta_\rho(b).
\end{equation}
Twisted fluctuations of the Dirac operator then take the form
\begin{equation}
	\Dir_{A_\rho}:= \Dir+A_\rho+\epsilon_1 JA_\rho J^{-1}\qquad \qquad \text{with} \qquad \qquad A_\rho\in \Omega^1_\Dir(\calA,\rho).
\end{equation}
We say that a twisted spectral triple is equipped with a $\bbZ_2$-symmetry $\Gamma$\footnote{We introduce a generalized definition of the operator $\Gamma$. Its precise geometric interpretation as either a grading or a $\bbZ_2$-symmetry is determined by its commutation relation with the Dirac operator.} if there exists a self-adjoint involution $\Gamma$ so that $[\Gamma, a]=0$ for any $a\in\calA$ and
 \begin{equation}
	\label{EqDefStructParamExt}
J \Gamma = \epsilon_2 \Gamma J,\qquad\qquad\quad\Dir\Gamma= \epsilon_3\Gamma\Dir, \qquad\qquad\qquad \epsilon_2, \epsilon_3\in \{\pm 1\}.
\end{equation}
The even case corresponds to $\epsilon_3=-1$, where $\Gamma$ is the standard $\bbZ_2$-grading on $\hs$. Conversely, in the absence of grading operator, the odd case will correspond to $\epsilon_3=1$.
 
More details on the origin of this construction can be found in \cite{TwistLandiMarti2016, TwistGaugeLandiMarti2018}.

We can introduce the $\rho$-product $\langle\, \cdot \, , \, \cdot \, \rangle_\rho$ defined in \cite{devastato2018lorentz} by the relation
\begin{align}
	\langle\psi , O\psi^\prime \rangle_\rho=\langle\rho(O)^\dagger\psi ,\psi^{\prime} \rangle_\rho\qquad\qquad \forall \psi, \psi^{\prime} \in \calH, \, O\in \calB(\calH).
\end{align}
We require $\rho$ to be $\calB(\calH)$-regular\footnote{i.e., regular on all $\calB(\calH)$. This stronger requirement is essential for many results.}. This implies the equality of left and right adjoints $\dagger_L$ and $\dagger_R$ (see \cite{nieuvsignchange} for the definitions) and ensures the uniqueness of the adjoint:
\begin{equation}
	(\, \cdot \,)^{+}:=(\, \cdot \,)^{\dagger_L}=(\, \cdot \,)^{\dagger_R}=\rho(\, \cdot \, )^{\dagger}.
\end{equation}
This leads to the notion of $\rho$-unitarity, defined by $U_\rho U_\rho^+=U_\rho^+ U_\rho=\bbbone$, with the corresponding set $\calU_\rho (\calB(\calH))$. As will be shown in Section~\ref{SecPseudoRiemST}, the orthochronous spin group provides such an example. The set of standard unitaries\footnote{Unitaries for the Hilbert product on $\calH$.} is denoted by $\calU(\calB(\calH))$.

The motivation for introducing $(\, \cdot \,)^{+}$ and $\calU_\rho (\calB(\calH))$ comes from the following observation. Using the adjoint action $Ad(a)\defeq aJaJ^{-1}$ for any $a\in\calA$, twisted fluctuations of the Dirac operator can be generated through elements $(u,u_\rho)\in (\calU(\algA),\calU_\rho (\algA))$ by
\begin{align}
	\Dir_{A_\rho} &:= Ad(u)\Dir (Ad(u))^+, \\
	\Dir_{A_\rho} &:= Ad(u_\rho)\Dir (Ad(u_\rho))^\dagger,
\end{align}
as discussed in \cite{TwistGaugeLandiMarti2018} and \cite{martinetti2024torsion}. In the following, we focus on the second way using $\rho$-unitaries, since this choice preserves the self-adjointness of the Dirac operator.

When $\rho\in\Inn(\calB(\calH))$, there exists a unitary operator $\tw\in \calU(\calB(\calH))$ on $\calH$ such that
\begin{equation}
	\rho(O)=\tw O \tw^\dagger\qquad\qquad\qquad \forall O \in\calB(\calH).\qquad\qquad\qquad
\end{equation}
In that case, the $\rho$-product can be expressed in terms of $\tw$ as
\[
\langle\, \cdot \, , \, \cdot \, \rangle_\tw := \langle\, \cdot \, , \tw \,\cdot \, \rangle.
\]
We will focus on this case in what follows, adopting the specific notations $(\, \cdot \,)^\Tadj:= \rho(\, \cdot \, )^\dagger$ for the adjoint and $U_\tw\in \calU_\tw(\calB(\calH))$ for the corresponding set of $\tw$-unitaries\footnote{This notation is more adapted as $\tw$ will be a central ingredient in the construction and that different operators $\tw$ can lead to the same twist $\rho$.}.  
If $\rho$ is a $*$-automorphism such that $\rho^2=\bbbone$, then $(\, \cdot \,)^\Tadj$ is an involution.  

\begin{proposition}[\cite{nieuvsignchange}]
	The $\calB(\calH)$-regularity condition implies that $\tw=\exp(i\theta)\tw^\dagger$.
\end{proposition}

The inner product $\inntw$ is Hermitian and indefinite if and only if $\tw=\tw^\dagger$, in which case it defines a Krein product. This result provides a direct and fundamental connection between twists and Krein products, establishing the first elements of the bridge between twisted spectral triples and pseudo-Riemannian (indefinite) structures.

In the following, we restrict ourselves to the case $\tw=\tw^\dagger$ and denote the corresponding twisted spectral triple by $(\calA,\, \calH,\, \Dir,\, J,\, \Gamma,\, \tw)$.

We now focus on the real structure and require the real operator $J$ to relate the left and right actions on the $\calA$-bimodule $\Omega^1_\Dir(\calA, \rho)$. The left and right actions then satisfy
\[
a \cdot (\, \cdot \, ) \cdot b = \rho(a)(\, \cdot \, )b = l(a)r(b)(\, \cdot \,),
\]
with $l(a)=Jr(a)^\dagger J^{-1}=\rho(a)$.

Let $\JH$ denote the real structure for usual spectral triples, i.e. relating left and right actions on the $\calA$-bimodule $\Omega^1_\Dir(\calA)$ 
\[
\hat{l}(a)\hat{r}(b)(\, \cdot \, ) = a \cdot (\, \cdot \, ) \cdot b = a(\, \cdot \,)b.
\]
\begin{proposition}[\cite{nieuviarts2025emergence}]
	The action $a \cdot (\, \cdot \, ) \cdot b=\rho(a)(\, \cdot \, ) b$ in \eqref{EqLR} is related to the real structure
	\begin{align}
		\label{EqTwistReal}
		J:= \tw \JH.\qquad\qquad
	\end{align}
\end{proposition}
Notably, this real structure also implements the adjoint action of $\tw$-unitary operators:
\begin{equation}
	\psi\,\to\, u_\tw \psi u_\tw^{-1}=u_\tw \psi u_\tw^\Tadj=u_\tw J u_\tw J^{-1}\psi\qquad\qquad (\psi, u_\tw) \in (\calH,\, \calU_\tw (\algA)).
\end{equation}
Such a definition of the real structure is therefore natural within the twisted framework.

Spectral triples are recovered in the particular case $\tw=\bbbone$, leading to $
[\Dir, a]_\rho\equiv[\Dir, a]$, $\inntw\equiv\inn$ with corresponding unitaries $u_\tw\equiv u$, and for the real structure $J=\tw\JH\equiv\JH$ that implement the usual adjoint action $\psi\,\to\, u \psi u^{-1}=u \psi u^\dagger=u \JH u \JH^{-1}\psi$.

From now on, we fix the relations between $\rho$ and $J,\Gamma$ by the following conditions
\begin{equation}
	\label{EqRelKJG}
	\tw J=\epsilon J\tw,\qquad\qquad\quad \tw \Gamma=\epsilon^\prime \Gamma\tw,\qquad\qquad\qquad  \epsilon, \epsilon^\prime\in \{\pm 1\}.
\end{equation}
This imply notably that $\rho=\rho^\circ$ on $\calA^\circ$. 

\section{Pseudo-Riemannian Spectral Triples}
\label{SecPseudoRiemST}
The following section offers an introduction to the so-called pseudo-Riemannian spectral triple. The proposed approach starts from the case of even-dimensional complex Clifford algebras, showing how a twist can be naturally defined in this framework and how it relates to both the orthochronous spin group and the Krein product underlying the pseudo-Riemannian structure. Further details can be found in \cite{nieuviarts2025emergence}.
 
Let $V\simeq \mathbb{R}^{2m}$ be a real vector space endowed with a metric $\g$ of signature $(n, 2m-n)$, inducing the splitting $V=V^+\oplus V^-$.  
For every $v\in V$, we define the operator $\cl(v)$, the complex linear representation of $v$, such that
\begin{equation}
	\cl(u)\cl(v) + \cl(v)\cl(u) = 2\g(u, v)\mbb \qquad \qquad \forall\, u, v \in V.
\end{equation}
Hence, each $\cl(v)$ is a $2^m \times 2^m$ complex matrix.  
We identify the elements $v\in V$ with their representation $\cl(v)$ from now on.  

The (complex) Clifford algebra $Cl(V, \g)$ is generated by $V$ together with the identity $\mbb$.  
In even dimensions, one has the isomorphism $Cl(V, \g)\simeq M_{2^m}(\mathbb{C})$, implying that $\Aut(Cl(V, \g))=\Inn(Cl(V, \g))$.  
An irreducible representation of $Cl(V, \g)$ is then given by the module $M=\mathbb{C}^{2^m}$.  
Equipped with a Hermitian inner product $\langle \, \cdot \, , \, \cdot \, \rangle$ whose adjoint is denoted by $\dagger$, this defines the Hilbert space $\calH:=(M,\langle \, \cdot \, , \, \cdot \, \rangle)$.
 
\begin{definition}[\cite{nieuviarts2025emergence}]
	The structural operators $\rho, \chi, \kappa$, called the twist, grading, and charge conjugation, are defined as elements of $\Aut(Cl(V, \g))$ such that for all $v \in V$:
	\begin{align}
		\label{EqDefRhoAll}
		&\rho(v):= v^{\dagger}, \qquad\qquad\quad \chi(v):= -v, \qquad\qquad\quad \kappa(v):= -\bar{v}.
	\end{align}
\end{definition}
Note that $\rho$ transforms into the adjoint only on $V$, so that $\rho$ is indeed an automorphism (not an anti-automorphism). The operators $\rho, \chi$, and $\kappa$ are involutions on $Cl(V, \g)$ and elements of $\Aut(V)$ satisfying $\rho\circ \chi=\chi \circ \rho$, $\rho\circ \kappa=\kappa \circ \rho$, and $\kappa\circ \chi=\chi \circ \kappa$.  
An important direct consequence of this definition is that the automorphism $\rho$ is $\mathcal{B}(\mathcal{H})$-regular.

Since $\Aut(Cl(V, \g))=\Inn(Cl(V, \g))$, there exist operators $\tw, \Gamma, C\in Cl(V, \g)$ such that
\begin{equation}
	\rho(\, \cdot \,)=\tw(\, \cdot \,)\tw^{-1},\qquad\quad \chi(\, \cdot \,)=\Gamma(\, \cdot \,)\Gamma^{-1},\qquad\quad \kappa(\, \cdot \,)=C(\, \cdot \,)C^{-1}.\qquad\qquad
\end{equation}
We require these operators to be unitary, ensuring that $\rho, \chi$, and $\kappa$ are $*$-automorphisms.

Let $\{e_a\}_{1\leq a\leq 2m}$ be a $\g$-orthonormal basis of $V$ and $\g_a\defeq \g(e_a, e_a)$. Requiring the $e_a$ to be unitaries implies that the eigenspaces induced by $\rho$ and $\g$ coincide: 
\begin{equation}
	\rho(e_a)=e_a^{\dagger}=\pm e_a=\g_a e_a,
\end{equation}
making $\rho$ be the parity operator itself. For positive-definite metrics, one has $\rho(e_a)=e_a$, corresponding to the case $\tw=\mbb$.

The real operator is defined as $J:= C\circ cc$.  
Let $\hat{V}$ be the vector space obtained from $V$ by a Wick rotation on the basis (see \cite{nieuviarts2025emergence}), corresponding to a positive-definite metric.  
The space $\hat{V}$ is associated with the charge conjugation $\hat{\kappa}$ and the real operator $\JH=\CH\circ cc$.

\begin{proposition}[\cite{nieuviarts2025emergence}]
	\label{PropJK}
	We have $\kappa=\hat{\kappa}\circ\rho=\rho\circ\hat{\kappa}$ on $Cl(V, \g)$.
\end{proposition}

The operators $C$ and $\hat{C}$ are then related by $C=\tw\hat{C}$, allowing us to recover the twisted real operator introduced in relation~\eqref{EqTwistReal}:
\begin{equation}
	\label{EqJPseudo}
	J=\tw\JH.
\end{equation}
This provides another fundamental connection between pseudo-Riemannian structures and twisted spectral triples.

The unit sphere is defined by $\US:=\{v\in V\, \mid \, \g(v,v)=\pm 1\}$, and the spin group by
\begin{equation}
	Spin(V, \g)=\{v_1\dots v_{2k}\, \mid\, v_1, \dots , v_{2k} \in \US\}.
\end{equation}
The orthochronous spin group $Spin(V, \g)^{+}\subset Spin(V, \g)$ consists of those elements that preserve the orientation of the subspaces $V^\pm$ (spatial and temporal). The Lorentz group is a special case corresponding to metrics $\g$ with signature $(1, 3)$. 

\begin{proposition}[\cite{nieuviarts2025emergence}]
	For all $x\in Spin(V, \g)^+$, one has $x^{-1}= \rho(x^{\dagger})$.
\end{proposition}

As a consequence, a $Spin(V, \g)^+$-invariant inner product is given by $\inntw=\langle \, \cdot \, , \tw\, \cdot \, \rangle$.  
Since $\rho$ is $\calB(\calH)$-regular, the adjoint $\Tadj$ is unique. The relation $x^{-1}= \rho(x^{\dagger})$ then becomes
\begin{align}
	x^{-1}= x^\Tadj.
\end{align}
Thus, every element $x\in Spin(V, \g)^+$ is a $\tw$-unitary operator, providing a concrete example of such operators. In the four-dimensional case, the Lorentz group is associated with an operator $\tw$ given by the fixed (temporal) gamma matrix $e_{(0)}$, i.e. the unique element $e_{(a)}$\footnote{The fixed basis $e_{(a)}$ is defined from the basis $e_{a}$, imposing it's elements to be invariant under $\glv$, see \cite{nieuviarts2025emergence} for more informations.} such that $\g_a=1$. The corresponding $\tw$-product defines the Krein structure. 

It is important to note that from the definition of the twist, $\tw$ remains fixed under the action of $Spin(V, \g)^+$, which justifies the notation $\inntw$.  
The relation $\rho(v)=v^\dagger$ then becomes $v=v^\Tadj$ for any $v\in V$, this property is preserved under the action of $Spin(V, \g)^+$.  

Note that many of the results in this section depend explicitly on the assumption that $V$ is even-dimensional. From now on, we require that $\tw=\tw^\dagger$ and $\Gamma=\Gamma^\dagger$.

\begin{proposition}[\cite{nieuviarts2025emergence}]\label{PropSignKJG}
	Taking $\epsilon=\pm 1$ and $\epsilon^\prime=\pm 1$ as functions of $n$, we recover the relations $\tw J=\epsilon J\tw$ and $\tw \Gamma=\epsilon^\prime \Gamma\tw$.
\end{proposition}

Let $\Man$ be a smooth $2m$-dimensional compact manifold endowed with a metric $\g$ of signature $(n, 2m-n)$.  
To any $v\in T_x\Man$ we associate the operator $\cl(v)$ such that
\begin{equation}
	\cl(u)\cl(v) + \cl(v)\cl(u) = 2\g_x(u, v)\mbb,\qquad \qquad \forall u, v \in T_x\Man.
\end{equation}
This allows us to define the Clifford algebra $Cl(T_x\Man, \g_x)$ and, globally, the Clifford bundle 
\[
Cl_{n,\, 2m-n}=\bigcup_{x\in\Man}Cl(T_x\Man, \g_x).
\]
The space of sections of the spinor bundle $\Sp$ is denoted by $\Gamma(\Sp)$ and is equipped with the Hermitian inner product 
\[
\inn : \Gamma(\Sp) \times \Gamma(\Sp)\to C^\infty(\Man, \mathbb{C})
\]
whose adjoint is denoted by $\dagger$.  
The corresponding fiber-wise inner product is written $\langle \psi, \psi^\prime\rangle_x:=\langle \psi(x), \psi^\prime(x)\rangle_x$, and we define $\langle \, \cdot \, , \, \cdot \, \rangle_{\twx}:=\langle \, \cdot \, , \twx\, \cdot \, \rangle_x$.

A $\g$-orthonormal basis of $T^*\Man$ is given by $\{\theta^1, \dots, \theta^{2m}\}$.  
The gamma matrices are the unitary operators $\gamma^a:= \cl(\theta^a)$, satisfying $\{\gamma^a, \gamma^b\}=2\g^{ab}\mbb$.  
The Dirac operator is then defined as the $\Tadj$-selfadjoint operator
\begin{equation}
	\KDir=i\gamma^\mu \nabla_\mu^{{\scriptscriptstyle\Sp}},\qquad\qquad
\end{equation}
where $\nabla_a^{{\scriptscriptstyle\Sp}}=\partial_a+\frac{1}{4}\Gamma^b_{\,\, a c}\gamma^c\gamma_b$ is the lift of the Levi-Civita connection to the spinor bundle.

The associated pseudo-Riemannian spectral triple is given by
\begin{equation}
(\CM,\, (L^2(\Man, \Sp),\, \inntw),\, \KDir,\, J,\, \Gamma).
\end{equation}

Let us now generalize this construction, starting from a vector space $E$ and a unitary operator $\tw=\tw^\dagger$, we define the Krein space $\calK:= (E, \langle \, \cdot \, , \, \cdot \, \rangle_\tw)$.

From \cite{barrett2007lorentzian, van2016krein}, the Dirac action is provided by the evaluation
\begin{equation}
	\label{EqrDirLagPRST}
\langle \psi , \KDir\psi \rangle_\tw.
\end{equation}
\begin{definition}[\cite{strohmaier2006noncommutative}]
	A pseudo-Riemannian spectral triple is a tuple $(\calA,\calK, \KDir)$ where $\calA$ is a unital involutive algebra acting as bounded operators on a Krein space $\calK$, equipped with a $\Tadj$-selfadjoint Dirac operator $\KDir$, so that $[\KDir ,a]$ is bounded $\forall a\in\calA$.
\end{definition}

For real pseudo-Riemannian spectral triple, extending the result of proposition \ref{PropJK} to the general case, we propose that the real operator $J$ is related to $\JH$ (as defined in section \ref{TwistedSec}) by the relation $J=\tw \JH$.
The algebraic constraints are specified by
\begin{equation}
	J^2 = \epsilon_0^\tw, \qquad\qquad\quad J \KDir = \epsilon_1^\tw \KDir J, \qquad\qquad\qquad \epsilon_0^\tw, \epsilon_1^\tw\in \{\pm 1\}.\qquad\quad
\end{equation}
Following the previous section, we say that a pseudo-Riemannian spectral triple is equipped with a $\bbZ_2$-symmetry $\Gamma$ if there exists a self-adjoint involution $\Gamma$ so that $[\Gamma, a]=0$ for any $a\in\calA$ and
\begin{equation}
	\label{EqDefStructParamExt}
	J \Gamma = \epsilon_2^\tw  \Gamma J,\qquad\qquad\quad \KDir\Gamma= \epsilon_3^\tw\Gamma\KDir, \qquad \qquad\qquad   \epsilon_2^\tw,  \epsilon_3^\tw \in \{\pm 1\}.
\end{equation}
 Following proposition \ref{PropSignKJG}, the relation between $\tw$ and $\Gamma, J$ will be fixed by the number $\epsilon, \epsilon^\prime=\pm 1$ through relation \eqref{EqRelKJG}.

Standard (Riemannian) spectral triples are recovered in the case $\tw=\bbbone$.
\newpage

\section{$\tw$-Morphism and Signature Change}
\label{SecMorphism}
The previous sections have presented two distinct extensions of the notion of a spectral triple. The only parameter for these extensions is the operator $\tw$, both the twisted and the pseudo-Riemannian spectral triples reduce to a spectral triple in the special case $\tw=\bbbone$. The following section introduces a morphism that allows one to relate these two extensions. We then show how this morphism acts as a signature change in the manifold case, and how it is implemented through a spacelike reflection.

The $\tw$-morphism $\twphi$ is defined as a bijective, involutive morphism parametrized by the operator $\tw$.  
It allows the transition 
\begin{equation}
(\calA,\, \calH,\, \Dir,\, J,\, \Gamma,\, \tw)\, \xleftrightarrow{\twphi}\,(\calA,\, \calK,\,\KDir,\, J,\, \Gamma)
\end{equation}
between twisted and pseudo-Riemannian spectral triples through the following action:
\begin{align}
	&\Dir\,\xrightarrow{\twphi}\, \KDir:= \tw \Dir,\\
	&[\Dir, a]_\rho\,\xrightarrow{\twphi}\, [\KDir, a]=\tw[\Dir, a]_\rho,\\
	&\calH:= (E,\, \langle \, \cdot \, , \, \cdot \, \rangle)\, \xrightarrow{\twphi}\, \calK:=  (E,\, \langle \, \cdot \, , \, \cdot \, \rangle_\tw= \langle\, \cdot \, , \tw \, \cdot \, \rangle).
\end{align}

With such a relation, the statements $\Dir^\dagger =\Dir$ and $(\KDir)^\Tadj=\KDir$ are equivalent, and the corresponding derivations are related by $[\Dir, a]_\rho =\tw [\KDir, a]$.
Similarly, the first-order conditions satisfy
\begin{equation}
[[\Dir, a]_\rho, b^{\circ} ]_{\rho^\circ}= \tw [[\KDir, a], b^{\circ}]=0.
\end{equation}
Ensuring their mutual implication.

The fluctuations are related by $\KDir_{A^\tw}=\tw\Dir_{A_\rho}$, where
\begin{align}
	\KDir_{A^\tw}=\KDir+A^\tw+\epsilon_1^\tw J A^{\,\tw} J^{-1}\qquad\text{and} \qquad \Dir_{A_\rho}=\Dir+A_\rho+ \epsilon_1 JA_\rho J^{-1}.\qquad
\end{align}

Fluctuations preserving the $\dagger$-selfadjointness of $\Dir$ and the $\Tadj$-selfadjointness of $\KDir$ are related through the $\tw$-unitary operator $U_\tw=u_\tw Ju_\tw J^{-1}$, with $u_\tw\in\calU_\tw (\algA)$
\begin{equation}
	\label{RelGaugeTF}
	\KDir_{A^\tw}:= U_\tw\KDir U_\tw^{\Tadj}= \tw V_\tw \Dir V_\tw^{\dagger} := \tw \Dir_{A_\rho},\qquad\qquad
\end{equation}
where $V_\tw=\rho(U_\tw)=\rho(u_\tw) J\rho(u _\tw) J^{-1}$ is also a $\tw$-unitary operator.
 
We then obtain the following relations between the sign parameters:
\begin{equation}
\epsilon_0=\epsilon_0^\tw, \qquad\qquad \epsilon_1^\tw=\epsilon\epsilon_1, \qquad\qquad \epsilon_2^\tw=\epsilon_2, \qquad\qquad \epsilon_3^\tw=\epsilon^\prime\epsilon_3.
\end{equation}
From now, we impose $\{\KDir, \Gamma\}=0$, so that $\Gamma$ is a grading for the pseudo-Riemannian spectral triples. This leads to the following relation for $(\calA,\, \calH,\, \Dir,\, J,\, \Gamma,\, \tw)$
\begin{equation}
	\{\Dir, \Gamma\}_\rho:= \Dir\Gamma+\rho(\Gamma)\Dir=\Dir\Gamma+\epsilon^\prime\Gamma\Dir=0.\qquad\qquad
\end{equation}
The $\tw$-morphism thus establishes a canonical correspondence between twisted and pseudo-Riemannian spectral triples. It preserves the essential structural properties, and connects generalized spectral triples through a bijective and involutive map. In this precise sense, the term morphism refers to a structure-preserving correspondence, which in particular preserves key spectral invariants such as the spectral and fermionic actions. The invariance of the real structure $J$ under the $\tw$-morphism is natural, as this operator acquires a consistent meaning in both frameworks, as shown by the equivalent equations~\eqref{EqTwistReal} and~\eqref{EqJPseudo}. Further details can be found in \cite{nieuvsignchange}.

Consider now the $2m$-dimensional compact pseudo-Riemannian manifold $(\Man,\g)$.  
Given $(\Man, \g)$, a spacelike reflection $r$ is defined to be a reflection\footnote{A reflection for $(\Man,\g)$ is an isometric automorphism $r$ on $\tm$ such that $r^2=\bbbone$.} with respect to the metric $\g$ such that the following metric is positive-definite on $\tm$
\begin{equation}
\gr(\, \cdot \,,\, \cdot \, ):= \g(\, \cdot \,,\,r \, \cdot\, ).
\end{equation}

Using the twist $\rho(\, \cdot \,)=\tw(\, \cdot \,)\tw$ defined in Section~\ref{SecPseudoRiemST}, we have for all $v\in \tm$:
\begin{equation}
	\rho(\cl(v))=\cl(rv)\qquad \qquad\qquad \text{where $r$ is spacelike}.\qquad\qquad
\end{equation}
That is $\rho$ is the parity operator (a spacelike coordinate reversal), consistent with the unitary cnhoice of Clifford basis, providing a clear geometric interpretation for the twist.

This relation implies that the metric relation $\gr(\, \cdot \,,\, \cdot \, ):= \g(\, \cdot \,,\,r \, \cdot\, )$ can be expressed at the level of the Clifford algebra as
\begin{equation}
	\label{EqRelMetCl}
	2\gr(u, v)\mbb=2\g(u, rv)\mbb=\{\cl(u), \cl(rv)\}=\{\cl(u), \rho(\cl(v))\}.
\end{equation}

Now, starting from the associated pseudo-Riemannian spectral triple $(\CM,\, \calK=(L^2(\Man, \Sp),\, \inntw),\, \KDir,\, J,\, \Gamma)$ with $\epsilon_3^\tw=-1$,  
we define $\Dir=\tw\KDir$ within the twisted spectral triple $(\CM,\, \calH,\,\Dir,\, J,\, \Gamma,\, \tw)$ obtained through the action of the $\tw$-morphism on $(\CM,\, \calK=(L^2(\Man, \Sp),\, \inntw),\, \KDir,\, J,\, \Gamma)$. 

In this particular case, we have $[\Dir, a]_\rho=[\Dir, a]$, as the twist $\rho$ acts trivially on $\CM$. We emphasize that $(\CM,\, \calH,\,\Dir,\, J,\, \Gamma,\, \tw)$ remains a twisted spectral triple by definition, since the twisted commutator notation is still applicable. The information of the twist remains of structural significance, as it is encoded within the operators $\Dir$ and $J$, which differ from those used in the standard spectral triple approach. The relation between $\rho$ and $\Dir$ will be highlighted in the following.

In NCG, the metric properties are encoded in the Connes distance formula (see \cite{connes1996gravity}), where the Clifford relations appear in the computation, encapsulating the metric.

\begin{proposition}[\cite{nieuvsignchange}]
	\label{PropDist}
	The twisted spectral triple is associated with the metric $\gr$.
\end{proposition}

\begin{proof}
	Using the Connes distance formula, the only modification in the twisted context arises at the level of the derivation:
	\begin{align}
		d(x,y)=\sup_{a\in \CM}\{ |a(x)-a(y)|\, \mid\, \| [\Dir, a]_\rho  \|\leq 1 \}.
	\end{align}
	Computing $\| [\Dir, a]_\rho  \|^2=\| [\Dir, a]  \|^2$ gives
	\begin{align}
		\| [\Dir, a]  \|^2&=\sup_{\substack{\psi\in\calH\\  \, \|\psi \| =1}}\{ \langle \left(\g(grad\, a, r grad\, {\bar a}) \right)\psi , \psi \rangle  \} 
	\end{align}
	where we recognize $\gr(grad\, a, grad\, {\bar a})$ so that $\| [\Dir, a]_\rho  \|^2= \|grad\, a\|_\infty^2$. This implies that $d(x,y)=d_{\gr}(x,y)$ i.e., that $\gr$ is the metric associated with the twisted spectral triple.
\end{proof}

The $\tw$-morphism therefore acts on the metric by changing its signature:
\begin{equation}
	\g(\, \cdot \,,\, \cdot \, )\,\xrightarrow{\twphi}\, \gr(\, \cdot \,,\, \cdot \, )= \g(\, \cdot \,,\,r\, \cdot\, ).\qquad
\end{equation}

The action on the Clifford representation is given by
\begin{equation}
	\label{EqtwActionCliff}
	\cl(u)\xrightarrow{\twphi}\tilde{\cl}(u)\defeq \tw\,\cl(u),
\end{equation}
where the representation $\tilde{\cl}$ is related to the metric $\gr$ through
\begin{equation}
	\label{EqDefCliffTw}
	2\gr(u,v)\mbb\ =\ \rho\,\!\big(\tilde{\cl}(u)\tilde{\cl}(v)\big)+\tilde{\cl}(v)\tilde{\cl}(u),
\end{equation}
which follows from equation~\eqref{EqRelMetCl} together with the result of proposition~\ref{PropDist}.  
This can be viewed as a twisted version of the usual Clifford relation, see \cite{nieuviarts2025emergence} for an abstract definition of twisted Clifford relations.

The metric can also be connected to the Clifford representation via the normalized trace of the product $\cl(u)\cl(v)$.  
The signature change induced by the $\tw$-morphism action~\eqref{EqtwActionCliff} is also manifest in this relation:
\begin{equation}
	\label{EqMetTrace}
	\g(u,v)=\frac{1}{2^m}\Tr\!\big(\cl(u)\cl(v)\big)\ \xrightarrow{\twphi}\ \frac{1}{2^m}\Tr\!\big(\tcl(u)\tcl(v)\big)=\frac{1}{2^m}\Tr\!\big(\cl(ru)\cl(v)\big)=\gr(u,v),
\end{equation}
showing that the $\tw$-morphism acts as a signature-changing transformation at the level of the Clifford representation. 

A Laplace-type operator with respect to $\gr$ is provided by
\begin{equation}
\frac{1}{2}(\rho(\Dir\Dir^\dagger)+\Dir\Dir^\dagger).
\end{equation}
This rewrite as
\begin{equation}
 \frac{1}{2}(\KDir(\KDir)^\dagger+(\KDir)^\dagger\KDir)= \frac{1}{2}((\tw\KDir)^2+(\KDir\tw)^2),
\end{equation}
where we recover the Laplace-type operator proposed in \cite{strohmaier2006noncommutative} on the rhs, an the usual laplacian formula when $\tw=\mbb$.

Let us now examine the effect of the $\tw$-morphism on the Christoffel symbols. Starting from $\KDir$, we define $\tgm^\mu=\tw\gamma^\mu$ so that $\Dir=-i\,\tgm^\mu\tilde\nabla_\mu^{\scriptscriptstyle R,\Sp}$, where
\begin{equation}
	\label{EqParticGamma}
	\tilde\nabla_\mu^{\scriptscriptstyle R,\Sp}=\partial_\mu+\frac{1}{4}\tilde\Gamma^{b}{}_{\mu a}\,\tgm^a\tgm_b
\end{equation}
is the associated spin connection, explicitly written in terms of the representation $\tilde{\cl}$.  
The coefficients $\tilde\Gamma^{b}{}_{\mu a}$ are defined from the $\Gamma^{b}{}_{\mu a}$ by the relation $\Gamma^{b}{}_{\mu a}=\g^{a}\g^{b}\tilde\Gamma^{b}{}_{\mu a}$.

The metric compatibility condition $\nabla_\nu\g_{\mu\kappa}=0$ for $\g$ implies
\[
\Gamma^\lambda{}_{\mu\nu}=\tfrac12\,\g^{\lambda\kappa}\big(\partial_\mu\g_{\nu\kappa}+\partial_\nu\g_{\mu\kappa}-\partial_\kappa\g_{\mu\nu}\big),
\]
and similarly for $\gr$, with $\Gamma^\lambda_{{\scriptscriptstyle R}\,\mu\nu}$ obtained from $\nabla^{\scriptscriptstyle R}_\nu\g^{\scriptscriptstyle R}_{\mu\kappa}=0$.

The reflected bases for $T^*\Man$ and $T\Man$ are defined by
\begin{equation}
	\label{EqReflect}
	dx^{r\mu}\defeq r\delta^\mu_\nu\,dx^\nu,\qquad\qquad \partial_{r\mu}\defeq r\delta_\mu^\nu\,\partial_\nu.
\end{equation}
If $\det(r)=-1$, the orientation of the basis is reversed.  
The map $dx^\mu\mapsto dx^{r\mu}$ is isometric for both $\g$ and $\gr$.  
Requiring integrability of $dx^{r\mu}$ yields coordinates $x^{r\mu}$ satisfying $dx^{r\mu}=d(x^{r\mu})$. In that case we have
\begin{align}
	dx^{r\mu}=\frac{\partial x^{r\mu}}{\partial x^\nu}dx^\nu=r\delta^\mu_\nu dx^\nu.
\end{align}
The dual basis is given by $\partial_{r\mu}=r\delta_\mu^\nu \partial_\nu$, and the corresponding Christoffel symbols satisfy
\begin{align}
	\nabla_\mu\partial_{r\nu}=\Gamma^{r\lambda}_{\,\, \mu r\nu}\partial_{r\lambda}.
\end{align}

\begin{proposition}[\cite{nieuviarts2025emergence}]
	The Christoffel symbols $\Gamma^\lambda_{\,\, \mu \nu}$ and $\Gamma^\lambda_{{{\scriptscriptstyle R}}\, \mu \nu}$ for $\g$ and $\gr$ are related by
	\begin{align}
		\label{RelatChristos}
		\Gamma^{r\lambda}_{\,\, \mu r\nu}=\Gamma^\lambda_{{{\scriptscriptstyle R}}\, \mu \nu}+\frac{1}{2}\gr^{\lambda\kappa}(\partial_{r\nu}\g_{\mu\kappa}-\partial_{\nu}\g^{\scriptscriptstyle R}_{\mu\kappa}).
	\end{align}
\end{proposition}

This result shows how the $\tw$-morphism acts on Christoffel symbols, connecting $\Gamma^\lambda_{\,\, \mu \nu}$ and $\Gamma^\lambda_{{{\scriptscriptstyle R}}\, \mu \nu}$ in a way that preserves their reality, in contrast with Wick rotations which introduce complex values during the signature change process.

The $\g$-orthonormal bases are also $\gr$-orthonormal, so the same vielbeins $e_a{}^{\mu}$ serve both metrics and relate the reflected bases by $dx^{ra}=e^{a}{}_{\mu}\,dx^{r\mu}$.  
Since $dx^{ra}=rdx^a=\g^{a}dx^a$, one has $\partial_{ra}=\delta^b_a\,\g_{b}\,\partial_b$.

From this and from relation~\eqref{EqDefCliffTw}, we deduce the generalized Clifford relation:
\begin{equation}
	\label{EqCliffGeneralisé}
	2\gr^{ab}\mbb\ =\ \tgm^a\tgm^b+s_{ab}\,\tgm^b\tgm^a,\qquad s_{ab}\defeq \g^{a}\g^{b}\in\{\pm1\}.
\end{equation}
Usual Clifford relations are recovered in the case $s_{ab}=1$.

As a consequence, in the reflected frame we have
\begin{equation}
	\label{EqRewritTGamma}
	\Gamma^{rb}{}_{\mu ra}=g^{b}g_{a}\,\Gamma^{b}{}_{\mu a}+g^{b}\partial_\mu g_{a}
	\ =\ \tilde{\Gamma}^{b}{}_{\mu a}+g^{b}\partial_\mu g_{a},
\end{equation}
while the vielbein relation extends to
\begin{equation}
	\label{EqRelatGammVielb}
	\Gamma^{rb}{}_{\mu ra}=e_a{}^{\nu}\,\Gamma^{r\lambda}{}_{\mu r\nu}\,e^b{}_{\lambda}-e_a{}^{\nu}\partial_\mu e^b{}_{\nu}.
\end{equation}
All these results lead to the following important result.

\begin{proposition}[\cite{nieuviarts2025emergence}]
	The Dirac operator in the twisted spectral triple takes the form
	\begin{equation}
		\Dir=-i \tgm^\mu\tilde\nabla_\mu^{{\scriptscriptstyle R, \Sp}}
	\end{equation}
	where $\tilde\nabla_\mu^{{\scriptscriptstyle R, \Sp}}$, the twisted version of the lift of  $\nabla_\mu^{{\scriptscriptstyle R}}$ to the spinor bundle, is given by
	\begin{equation}
		\label{EqDefDir}
		\tilde\nabla_\mu^{{\scriptscriptstyle R, \Sp}}=\partial_\mu+\frac{1}{4}\Gamma^b_{{{\scriptscriptstyle R}}\, \mu a}\tgm^a\tgm_b+\frac{1}{4}K^b_{\,\, \mu a}\tgm^a\tgm_b
	\end{equation}
	where $K^b_{\,\, \mu a}=g^b(\frac{1}{2}\g^{b\kappa}(\partial_{ra}\g_{\mu\kappa}-\partial_{a}\g^{\scriptscriptstyle R}_{\mu\kappa})-\partial_\mu g_a)$.
\end{proposition}

This provides an intrinsic expression of $\Dir$ in terms of the $\tgm^a$, the Christoffel symbols $\Gamma^{b}_{{\scriptscriptstyle R}\,\mu a}$, and the terms $K^{b}{}_{\mu a}$, each directly associated with the metric $\gr$ and $r$. The $\tw$-morphism then acts as a local signature change\footnote{Because the action of $r$ is local.}, analogous to a Wick rotation but without introducing any complex structure. Since the construction is carried out in the compact setting, this transformation should be understood at the level of the Clifford and Dirac data, rather than as a full pseudo-Riemannian space-time construction endowed with a global causal structure. Further details can be found in \cite{nieuviarts2025emergence}.
 \newpage

\section{Almost-Commutative Twisted Spectral Triples}
\label{SecACTW}

 Section \ref{SecMorphism}  presents how starting from a pseudo-Riemannian spectral triple, the $\tw$-morphism provide a twisted spectral triple associated with a Riemannian metric. The present section proposes a reversed "top-down" approach where a twisted spectral triple is taken as the primary object, on which we build an almost-commutative spectral triple structure. This section contains the main result of this proceeding. It shows how a pseudo-Riemannian spectral triple structure emerges within this almost-commutative twisted spectral triple. We focus on twisted spectral triples whose axioms correspond, through the $\tw$-morphism, to those of a pseudo-Riemannian spectral triple of KO-dimension~6, motivated by the fact that this correspond to our space-time case\footnote{Taking a metric with signature $(n, 2m-n)$, the KO dimension is the number $2(n-m) \mod  8$, which give a KO dimension $6$ in signature $(1, 3)$. More informations can be found in \cite{nieuviarts2025emergence}.}. The finite part also has KO-dimension~6, the motivation for this choice is discussed in \cite{barrett2007lorentzian}, in connection with the issues of fermion doubling and metric signature in the noncommutative Standard Model.
 
In KO-dimension $6$, the axioms of any even real spectral triple, Riemannian or pseudo-Riemannian, are encoded by the signs $\tilde{\epsilon}_0=1, \tilde{\epsilon}_1=1,\tilde{\epsilon}_2=-1, \tilde{\epsilon}_3=-1$. Motivated by the $\tw$-morphism correspondence, but without assuming any underlying pseudo-Riemannian spectral triple data, we consider twisted spectral triples $(\CM,\, \calH,\, \Dir,\, J,\, \Gamma,\, \tw)$  whose sign data match these axioms at the purely algebraic level
\begin{equation}
	\epsilon_0=\tilde{\epsilon}_0=1,\quad\qquad \epsilon_1=\epsilon\tilde{\epsilon}_1=\epsilon,\qquad\quad \epsilon_2=\tilde{\epsilon}_2=-1,\qquad\quad \epsilon_3=\epsilon^\prime\tilde{\epsilon}_3=-\epsilon^\prime,
\end{equation}
where the parameters $\epsilon$ and $\epsilon^\prime$ are left unspecified. The finite part is given by the finite-dimensional algebra $\AF$, within the KO-dimension~6 spectral triple $(\AF,\, \HF,\, \FDir,\, \JF,\,\GF)$, in the spirit of the noncommutative Standard Model. Defining $\calA_p:=\CM\otimes\AF$, the product spectral triple is $
(\calA_p,\, \calH_p,\, \Dir_p,\, J_p, \,\Gamma_p,\, \tw_p)$
with elements given by
\begin{equation}
	\calH_p:=\calH\otimes \HF\qquad\qquad J_p:= J\otimes \JF\qquad\qquad \Gamma_p:= \Gamma\otimes \GF \qquad\qquad \rho_{p}:=\rho\otimes \fbb.
\end{equation}
For the Dirac operator, we propose the following general form
\begin{equation}
\Dir_p:= \Dir \otimes \fbb + \calO\otimes \FDir, 
\end{equation} 
where $\calO$ is an operator that will be specified below. The square of $\Dir_p$ gives
\begin{equation}
\Dir_p^2=\Dir^2 \otimes \fbb+ \calO^2\otimes \FDir^2+\{\calO, \Dir\}\otimes \FDir.
\end{equation}
In Almost-commutative geometry, the usual approach imposes $\calO^2=\mbb$ and $\{\calO,\Dir\}=0$ to recover the additive relation $\Dir_p^2=\Dir^2\otimes \fbb+\mbb\otimes \FDir^2$,
which is the spectral counterpart of a metric product, see \cite{vanhecke1999product}. In the present setting, we do not impose $\{\calO,\Dir\}=0$ \emph{a priori}. The product remains topologically cartesian through the algebra $\CM\otimes\AF$, while its metric content is no longer additive in the usual sense.

 The algebraic constraints are
\begin{equation}
	J_p^2=\epsilon_0^p\quad\quad\quad J_p\Dir_p =\epsilon_1^p\Dir_p J_p\quad\quad\quad J_p\Gamma_p =\epsilon_2^p\Gamma_p J_p \quad\quad\quad \Gamma_p \Dir_p = \epsilon_3^p\Dir_p \Gamma_p.\quad \quad \quad
\end{equation}
\begin{proposition}[\cite{nieuviarts2025emergence}]
	The spectral triple algebraic constraints imply that:
	\begin{equation}
		\label{ConstrainAlgST}
		\calO=\calO^\dagger\qquad\quad \qquad J\calO=\epsilon\calO J \qquad \qquad \quad \Gamma\calO=\epsilon^\prime\calO \Gamma.\qquad\qquad
	\end{equation}
\end{proposition}
We obtain the following expected twisted anticommutation relation $\{\Dir_p, \Gamma_p\}_{\rho_p}=0.$

In the $\epsilon=-1,\epsilon^\prime=1$ case, one may choose $\calO=\Gamma$ so that $\{\calO,\Dir\}=0$. We shall not investigate this special case here. More generally, a noteworthy feature is that relations \eqref{ConstrainAlgST} together with the requirement $\calO^2=\mbb$ coincide precisely with the algebraic constraints the fundamental symmetry $\tw$ must satisfy. This nontrivial fact provides a natural motivation to consider the following Dirac operator
\begin{equation}
	\label{EqDirTot}
	\Dir_p=\Dir\otimes \fbb + \tw\otimes \FDir.
\end{equation}

The operator $\tw$ is thus associated with the finite part and becomes a central ingredient of the construction. Note that, from this perspective, it is the algebraic constraints of the structure that force the operator $\calO$ to be a fundamental symmetry. 

This specific choice also permit to satisfy the full first order condition
\begin{equation}
[[\Dir_p, a]_{\rho_p}, b_p^\circ]_{\rho_p^\circ} = 0, \qquad\qquad \forall a, b \in \calA_p.
\end{equation}
\begin{proposition}[\cite{nieuviarts2025emergence}]
	The choice $\calO=\tw$ ensures the existence of independent derivations:
	\begin{align}
		[\Dir_p, a]_{\rho_p}=[\Dir , a_1]_{\rho}\otimes a_2+ \tw a_1\otimes [\FDir , a_2].
	\end{align}
\end{proposition}
These results remain valid for noncommutative extensions of the manifold sector, where $[\Dir , a_1]_{\rho} \neq [\Dir , a_1]$. For this reason, we keep the notation $[\Dir , a_1]_{\rho}$ to emphasize the central role of the twisted commutator in the algebraic structure. Other choices for $\calO$ would not permit the coexistence of the derivations $[\FDir , a_2]$ and $[\Dir , a_1]_{\rho}$ within the full derivation. A similar result can be obtained in the special case $\epsilon=-1,\epsilon^\prime=1$, taking $\calO=\Gamma$ since $\Gamma$ commutes both with $\tw$ and $\calA$. This result offers a new perspective on the origin of twisted spectral triples from the almost-commutative viewpoint.

We can then rewrite the operator $\Dir_p$ as
\begin{equation}
	\Dir_p:=\tw \KDir_p=\tw(\Dir^\tw\otimes \fbb + \mbb\otimes \FDir).
\end{equation}

This extends to the fluctuations of $\Dir_p$ generated by the $\tw\otimes\fbb$-unitary operators:
\begin{equation}
	\label{EqGaugeTf}
	(U_\tw\otimes U) \Dir_p (U_\tw^\dagger \otimes U^{\dagger_F})= \tw(\Dir_{ A^\tw}^\tw \otimes \fbb + \mbb\otimes \Dir_{A_F}),
\end{equation}
where $U_\tw=u_\tw Ju_\tw J^{-1}$ is a $\tw$-unitary operator on $\calH$, with $U=u \JF u \JF^{-1}$ a unitary operator on $\HF$, and $\dagger_F$ the associated adjoint.
 
We now turn to the inner product structure. For any $\psi=\psi_1\otimes \psi_2$ and $\psi^\prime=\psi^\prime_1\otimes \psi^\prime_2$ in $\calH_p$, the inner product on $\calH_p$ is given by
\begin{equation}
	\langle \psi , \psi^\prime \rangle_p:=\langle \psi_1 ,\psi^\prime_1 \rangle \langle \psi_2 ,\psi^\prime_2 \rangle_F.\qquad\qquad
\end{equation}
The $\tw\otimes\fbb$-product is defined by 
\begin{equation}
	\label{EqFermAct}
	\langle \, \cdot \, , \, \cdot \, \rangle_{\tw\otimes \fbb}:= \langle \, \cdot \, , \tw\otimes \fbb\, \cdot \, \rangle_p =\langle \, \cdot \, , \, \cdot \, \rangle_{\tw}\langle \, \cdot \, , \, \cdot \, \rangle_{F}.
\end{equation}
We propose the fermionic action\footnote{The fermionic action in the noncommutative approach to the Standard Model yields the Euclidean Dirac action. Our formula differ from the one of \cite{chamseddine2007gravity}, which contain the real structure within the product.} to be the evaluation of $\Dir_p$ by $\langle \, .\,  , \, .\,  \rangle_p$
\begin{equation}
	\langle \psi , \Dir_p\psi^\prime \rangle_p.
\end{equation}
Since $\tw$ is a fixed structure for the twisted spectral triple, relation \eqref{EqFermAct} rewrites as
\begin{equation}
	\label{EqEval}
	\langle \psi , \KDir_p\psi^\prime \rangle_{\tw\otimes \fbb}=\langle \psi_1 , \Dir^\tw\psi^\prime_1 \rangle_{\tw} \langle \psi_2 ,\psi^\prime_2 \rangle_F+\langle \psi_1 , \psi^\prime_1 \rangle_{ \tw} \langle \psi_2 ,\FDir \psi^\prime_2 \rangle_F.
\end{equation}
We recover the (pseudo-Riemannian) Dirac action defined in \eqref{EqrDirLagPRST}. It is worth noting that the fixed nature of $\tw$ arises here from its association with the finite part in \eqref{EqDirTot}, and was independently deduced from the definition of the twist $\rho$ in equation \eqref{EqDefRhoAll}. The fact that $\tw$ is fixed singles out the KO-dimension $6$ pseudo-Riemannian spectral triple 
$(\CM,\, \calK,\,\KDir,\, J,\, \Gamma)$ among the symmetries allowed by the invariance 
of the fermionic action, permitting only 
pseudo-Riemannian-type structures and symmetries.

This action is invariant under the following gauge transformation
\begin{equation}
\KDir_p\,\to\, (U_\tw\otimes U) \KDir_p (U_\tw^\Tadj \otimes U^{\dagger_F}), \qquad\qquad \psi\,\to\,  (U_\tw\otimes U) \psi,
\end{equation}
equivalent to the action generating the fluctuations in \eqref{EqGaugeTf}, using equation \eqref{RelGaugeTF}.

 Following \cite{barrett2007lorentzian}, the fermion doubling problem is solved by working on the subspace of $\calH_p$ whose element $\psi$ respects the following relation\footnote{Note that the difference between this formula and the one of \cite{barrett2007lorentzian} stems from the differing conventions regarding the square of the grading operator; specifically, $\Gamma^2=-1$ in \cite{barrett2007lorentzian} whereas $\Gamma^2=1$ here. The two gradings are related by the complex coeficient $-i$. This also results in a different relation between $J$ and $\Gamma$.}
 \begin{equation}
 \Gamma_p\psi=J_p\psi=\psi.
 \end{equation}
As a consequence we obtain the following symmetry relation
\begin{equation}
	\langle \psi , \Dir_p\psi^\prime \rangle_p=	\langle J_p\psi , \Dir_p\psi^\prime \rangle_p=\epsilon_0^p\langle   J_p \Dir_p\psi^\prime, \psi \rangle_p=	\epsilon\langle \psi^\prime , \Dir_p\psi \rangle_p, 
\end{equation}
using the fact that $\epsilon_0^p=1$, $\epsilon_1^p=\epsilon$, see \cite{nieuviarts2025emergence}. To obtain a non vanishing action $\langle \psi , \Dir_p\psi \rangle_p$, the symmetry requirement on the fermionic action then implies the necessity of using anticommuting Grassmann variables for the fermionic fields only in the $\epsilon = -1$ case.

In this top-down approach, following the results of section \ref{SecMorphism}, we require the operator $\Dir$ to be associated with a positive metric $\gr$ through the twisted Clifford relation \eqref{EqDefCliffTw}. We define $\KDir=\tw\Dir$, then the matrices $\kgm^\mu\defeq \tw\tgm^\mu$, connected to the metric $\gk$ by the usual Clifford relation $2\gk^{ab}\mbb := \kgm^a\kgm^b+\kgm^b\kgm^a$.
The signature of $\gk$ is determined by the choice of $\tw$, which is constrained by the structure of the almost-commutative twisted spectral triple through the parameters $\epsilon$ and $\epsilon^\prime$:
\begin{equation}
	\tw^\dagger=\tw,\qquad \qquad\qquad \tw \Gamma=\epsilon^\prime\Gamma\tw, \qquad \qquad\qquad \tw J=\epsilon J\tw.\qquad\qquad
\end{equation}

\begin{proposition}[\cite{nieuviarts2025emergence}]
	In dimension~4, the operator $\tw$ satisfying the above conditions induces the transformation from the metric $\gr$ with signature $(+, +, +, +)$ to the metric $\gk$ with signature $(+, -, -, -)$ for $\epsilon=-1$, and to $(+, +, +, -)$ for $\epsilon=1$.
\end{proposition}

The Lorentzian quantum field theory case, corresponding to $\tw=\kgm^{(0)}$\footnote{The notation $\kgm^{(0)}$ follows from the results of section \ref{SecPseudoRiemST}. This emphasizes the fixed nature of $\tw$, in contrast with the gamma matrix $\kgm^{0}$, which may transform under coordinate or Lorentz transformations in suitable pictures. See \cite{nieuviarts2025emergence} for further details.}, is recovered for $\epsilon=-1$, consistently with the physical requirement that fermionic fields be described by anticommuting Grassmann variables. The term $\tw\otimes \FDir$ then yields the expected Dirac mass term $\kgm^{(0)}m$ in the fermionic action~\eqref{EqEval}, while the term $\langle \psi , \Dir^\tw\psi \rangle_{\tw}$ provides the kinetic contribution. Further results can be found in \cite{nieuviarts2025emergence}.

We emphasize that twisted spectral triple structures are parametrized by the choice of $\tw$, itself constrained by the algebraic structure. There is no hidden pseudo-Riemannian datum in this construction: defining the axioms through the $\epsilon_i$'s from the  $\tilde{\epsilon}_i$'s does not encode any underlying pseudo-Riemannian geometry, since these sign data apply equally well to both Riemannian and pseudo-Riemannian spectral triple structures. 

These results provide new insights into the origin of pseudo-Riemannian signatures and the associated Krein structures, which appear here as natural consequences of the almost-commutative structure underlying twisted spectral triples. In four dimensions, the presented top-down approach singles out the Lorentzian case from the algebraic constraints of the almost-commutative spectral triple. Since the construction is developed in the compact setting, this result should however be interpreted at a local or algebraic level, rather than as a full Lorentzian space-time construction with global causal structure.

An important structural aspect concerns the mixed term $\{\calO,\Dir\}\otimes \FDir$ in $\Dir_p^2$. It shows that, although the construction remains topologically Cartesian through the product algebra $\CM\otimes\AF$, its metric content goes beyond the usual metric product picture. The manifold and finite sectors do not simply carry two independent metric structures; rather, their metric data enter a single spectral object in a genuinely intertwined way. This interpretation is reinforced by the fact that the pseudo-Riemannian and the finite sectors both carry KO-dimension $6$. This suggests that the almost-commutative geometry should here be viewed not as a mere tensor product of two separate geometries, but as a more unified object with a global and mixed metric content.

The main point is therefore the following: the noncommutative extension at the heart of the noncommutative Standard Model may itself contribute to the emergence of time from a purely Riemannian background. In this perspective, time emerges as a consequence of the geometrization of fundamental interactions.

The present results should be viewed as a first step toward a broader picture. It remains in particular to extend the analysis to other KO-dimensions, to investigate the case of odd-dimensional manifolds, and to better understand the full implications of this non-product metric structure. A further question is whether the twist considered here admits a modular interpretation, in the spirit of \cite{connes2006type} in which twisted spectral triples were first introduced. An especially important direction for future work will be the study of the spectral action, together with its symmetries and physical implications.

\bibliography{bibliography}
\end{document}